\documentclass[12pt]{article}
\usepackage[utf8x]{inputenc}
\usepackage{graphicx}
\usepackage{enumerate}
\usepackage{units}
\usepackage{multicol}
\usepackage{amsmath}
\usepackage{caption}
\usepackage{subcaption}
\usepackage{parskip}
\usepackage[english]{babel}
\usepackage{amsthm}
\usepackage{amssymb}
\usepackage{MnSymbol,wasysym}
\usepackage{siunitx}
\usepackage[ruled,vlined]{algorithm2e}
\usepackage[top=1in, bottom=1in, left=1in, right=1in]{geometry}

\usepackage{verbatim}
\usepackage{color}
\usepackage{xcolor}
\usepackage{mathtools}
\usepackage{stmaryrd}
\usepackage{listings}
\usepackage{natbib}
\usepackage{mathtools}
\mathtoolsset{showonlyrefs}
\newcommand{\E}{{\rm I\kern-.3em E}}

\begingroup
    \makeatletter
    \@for\theoremstyle:=definition,remark,plain\do{%
        \expandafter\g@addto@macro\csname th@\theoremstyle\endcsname{%
            \addtolength\thm@preskip\parskip
            }%
        }
\endgroup

\DeclareFixedFont{\ttb}{T1}{txtt}{bx}{n}{12} % for bold
\DeclareFixedFont{\ttm}{T1}{txtt}{m}{n}{12}  % for normal

\definecolor{mygreen}{rgb}{0,0.6,0}
\definecolor{myblue}{rgb}{0.09,0.57,0.98}

\definecolor{deepgreen}{rgb}{0,0.4,0}

\lstdefinestyle{python}{
	language=python,
	basicstyle=\color{black}\ttfamily\footnotesize,
	stringstyle=\color{deepgreen}\slshape,
	commentstyle=\color{gray}\slshape,
	keywordstyle=\color{red}\bf,
	emphstyle=\color{blue}\bf,
	tabsize=2,
	showstringspaces=false,
	emph={access,and,break,class,continue,def,del,elif,else,except,exec,finally,for,from,global,if,import,in,i s,lambda,not,or,pass,print,raise,return,try,while,as},
	upquote=true,
	morecomment=[s]{"""}{"""},
	literate=*
	{:}{{\textcolor{blue}:}}{1}%
	{/}{{\textcolor{blue}/}}{1}%
	{*}{{\textcolor{blue}*}}{1}%
	{!}{{\textcolor{blue}!}}{1}%
	aboveskip=\baselineskip,
	xleftmargin=20pt, xrightmargin=15pt,
	frame=lines,
	numbers=left, numberstyle=\tiny
}
\lstnewenvironment{python}{\lstset{style=python}}{}

\newtheorem{theorem}{Theorem}
\newtheorem{lemma}{Lemma}
\newtheoremstyle{component}{}{}{}{}{\it}{.}{.5em}{\thmnote{#3}#1}
\theoremstyle{component}

\title{Cache-Friendly Search Trees; \\or, In Which Everything Beats {\tt std::set}}
\author{Jeffrey Barratt\\\texttt{jbarratt} \and Brian Zhang\\\texttt{bhz}}
\date{\today}

\begin{document}

\maketitle

% %http://web.stanford.edu/class/cs166/handouts/090%20Final%20Project%20Requirements.pdf << spec

\section{Introduction}

While a lot of work in theoretical computer science has gone into optimizing the runtime and space usage of data structures, such work very often neglects a very important component of modern computers: the cache. In doing so, very often, data structures are developed that achieve theoretically-good runtimes but are slow in practice due to a large number of cache misses. In 1999, \cite{frigo1999cache} introduced the notion of a {\it cache-oblivious algorithm}: an algorithm that uses the cache to its advantage, {\it regardless of the size or structure of said cache}. Since then, various authors have designed cache-oblivious algorithms and data structures for problems from matrix multiplication to array sorting \citep{demaine2002cache}. We focus in this work on cache-oblivious search trees; i.e. implementing an ordered dictionary in a cache-friendly manner. We will start by presenting an overview of cache-oblivious data structures, especially cache-oblivious search trees. We then give practical results using these cache-oblivious structures on modern-day machinery, comparing them to the standard {\tt std::set} and other cache-friendly dictionaries such as B-trees.

\section{The Ideal-Cache Model}

To discuss caches theoretically, we first need to give a theoretical model that makes use of a cache. The {\it ideal-cache model} was first introduced by \cite{frigo1999cache}. In this model, a computer's memory is modeled as a two-level hierarchy with a {\it disk} and a {\it cache}. The cache has size $M$ and is split into {\it blocks} of size $B$ each\footnote{The letter $B$ will conflict later on with the $B$ in $B$-tree. This will be annoying, but we continue anyway and hope the reader is not too confused.}. The cache is assumed to be:
\begin{enumerate}[(1)]
    \item {\it fully-associative}: any block in the disk can be moved into any block of the cache.
    \item {\it ideal}: the cache ``knows'' which blocks of the disk will be accessed in the future, and thus always evicts blocks from the cache optimally.
    \item {\it tall}: $M = \Omega(B^2)$.
\end{enumerate}
The running time of an algorithm is then measured by the number of {\it memory transfers} that it must perform, with the idea that, in practice, transferring blocks of memory is significantly slower than actually running an algorithm that counting memory transfers is a reasonable metric. Some justification of the assumptions is necessary---in practice, caches do not know the future; instead, a heuristic is used. Two common heuristics used are the LRU (Least Recently Used) and FIFO (First-In First-Out) heuristics: in the former, the least-recently accessed block is evicted when the cache is full; in the latter, the least-recently added block is evicted. Further, practical caches are not fully associative. In the other extreme, a {\it 1-way associative cache} is one in which each block of the disk can only appear in one location in the cache. At first glance, neither of these seem close to the ideal-cache model. However, \cite{demaine2002cache} gives the following results:

\begin{theorem}[\cite{demaine2002cache}, Corollary 1] Suppose that the number of memory transfers depends only polynomially on the size $M$ of the cache; in particular, that halving $M$ will only increase the complexity by a constant factor. Then an LRU or FIFO cache only uses a constant factor more memory transfers than an ideal cache.
\end{theorem}
\begin{theorem}[\cite{demaine2002cache}, Lemma 2]
For some $\alpha > 0$, a fully associative LRU cache of size $\alpha M$ and block size $B$ can be simulated in $M$ space with $O(1)$ expected memory access time to any block.
\end{theorem}
In particular, such a cache would fit completely into our ideal cache of size $M$ and be 1-way associative (since we have $O(1)$ expected memory access); thus, assumption (1) is reasonable.

\cite{frigo1999cache} also show that an algorithm that is optimal in this two-level ideal cache model is also optimal in a memory hierarchy with more than two levels of caches. For simplicity, we will omit the formalism and refer readers to the original paper for the details of this result. 

The results above, combined, show the power of the ideal-cache model: an algorithm that is optimal in the ideal-cache model is also optimal (to within a constant factor) in any LRU memory hierarchy.

\section{Cache-Oblivious Model and Algorithms}

The cache-oblivious model was also first introduced by \cite{frigo1999cache}. The idea behind the model is simple: assume we do not know the block size $B$ or cache size $M$ of our computer and want to extend our algorithms to work on arbitrary architectures. We perform analyses on our algorithms using the above ideal-cache model and an unspecified block size $B$, which will become a variable in our memory transfer equations. The overall goal of designing cache-oblivious algorithms is to match the time bound of an algorithm designed with the block size $B$ in mind; for instance, a B-tree with nodes limited to size $B - 1$ and $2B - 1$, a constant factor times the size of $B$, results in ideal bounds on memory transfers to complete relevant operations. We will go deeper into this analysis with examples later in this section as well as in the following sections.

We will begin by exploring the process of scanning an array in Section \ref{subsec:arraytraversal} followed by an attempt at binary search with sub-optimal memory transfer bounds in Section \ref{subsec:failedbs}. This section will give a point of reference for analyses later in the paper.

\subsection{Cache-Oblivious Analysis of Array Traversal} \label{subsec:arraytraversal}

While this is a simple problem with a simple solution which performs optimally regardless of the cache size, the analysis of this algorithm is useful. A traversal of a contiguous block of memory, usually expressed as an array, is an example of a cache-oblivious algorithm; it performs optimally regardless of the cache size and structure. We will now prove bounds on the memory usage of this simple case.

\begin{theorem}[\cite{demaine2002cache}, Theorem 1]
Scanning $N$ elements stored in a contiguous segment of memory costs at most $\lceil N/B \rceil + 1$ memory transfers.
\end{theorem}

The proof of this theorem is obvious; to read all the elements in the array, we must scan each element, which are held in $N$ separate memory positions. However, since these memory positions are adjacent in memory, we can access blocks of $B$ elements with only one memory transfer, meaning that we only need $\lceil N/B \rceil + 1$ memory transfers; the 1 comes from overflow of the end of the array onto the next block.

Even with this trivial example, we see the power of designing cache-oblivious algorithms: with the same data structure or algorithm, we can produce an algorithm that uses the asymptotically-minimal number of memory transfers to accomplish a task without tuning parameters and regardless of the block size $B$.

\subsection{Cache-Oblivious Analysis of Binary Search} \label{subsec:failedbs}

When we perform a standard binary search on a sorted array, we can find elements in the optimal $O(\log n)$ runtime. However, performing optimally with the number of comparisons or computer instructions does not necessarily guarantee optimal memory usage.

\begin{theorem}[\cite{demaine2002cache}, Theorem 3]
Binary search on a sorted array incurs $O(\log N - \log B)$ memory transfers.
\end{theorem}

The proof of this theorem is again easy to see; imagine splitting the sorted array into blocks of size $B$. We access $\log N$ elements during our search, with the distance between element accesses within the array halving each time. For the final $\log B$ elements, all accesses in the search will be made within that same block. Thus we need $O(\log N - \log B)$ memory transfers to complete a binary search of a sorted array.

As we will see later in the paper, the lower bound of finding an element in a sorted set is $O(\log_B N)$, meaning that the naive binary search algorithm is suboptimal in its total number of memory transfers. We will improve on the naive binary search asymptotic memory transfer cost later in this paper.

\section{Cache-Oblivious Search Trees}
In this section, we will present a design of a cache-oblivious search tree, as given by \cite{prokop1999cache} and \cite{brodal2002cache}.
\subsection{Static Search Trees: The van Emde Boas Layout}
We begin by supposing that we had access to all the elements that we wish to put in our binary search tree. In this section, we will present a binary search tree layout that achieves an optimal average-case number of memory transfers. This structure was first given by \cite{prokop1999cache} as part of his master's thesis, but he apparently considered the proofs so trivial that they were omitted.  We thus present the proofs as given by \cite{demaine2002cache}. 
\begin{theorem}[\cite{demaine2002cache}, Theorem 7]
In a search tree of size $N$, starting from an empty cache, $\Omega(\log_B N)$ memory transfers are required to search for an element.
\end{theorem}
\begin{proof}
The sorted position of one of $N$ elements contains $\Theta(\log N)$ bits of information. A single block of size $B$ can be used to recover only $\Theta(\log B)$ bits of information; in particular, the sorted location of our desired element among the $\Theta(B)$ elements in the block. Thus, to recover the $\Theta(\log N)$ bits of information required to determine a given element's location in the sorted order, we require $\Omega(\log N / \log B) = \Omega(\log_B N)$ memory transfers.
\end{proof}

Assume WLOG that $N$ is a power of 2; if not, pad the list of elements up to a power of 2. This incurs at most constant-factor loss. The {\it van Emde Boas} (vEB) layout\footnote{As far as we are aware, the van Emde Boas layout was not discovered by van Emde Boas. The name, to our knowledge, comes because of its similarity to the {\it van Emde Boas tree}, which looks similar but has nothing to do with cache locality.} for a complete binary search on $N = 2^H$ nodes is defined recursively as follows.
\begin{enumerate}[(1)]
\item If $H=1$ then there is only one node. Put it in memory. 
\item If $H > 1$ then let $2^\ell$ be the largest power of 2 strictly less than $H$, and let $m = H - 2^\ell$. Lay out the top $m$ layers of the tree in the van Emde Boas layout, followed, in order, by the $2^m$ subtrees of the top $m$ layers (each of which has $2^{2^\ell} - 1$ nodes), also in vEB layout.
\end{enumerate}
A visualization of the vEB layout is given in \cite{demaine2002cache} on page 16. Notice that, not counting the root node, a vEB layout stores a ``tree block'' of size $B$ (i.e. induced subtrees of a subset of nodes) contiguously in memory, except at the root. 
\begin{theorem}[\cite{demaine2002cache}, Theorem 8]
Searching in a tree in vEB layout takes $O(\log_B N)$ memory transfers.
\end{theorem}
\begin{proof}
After the root, reading one block of memory will allow one to descend through $\Theta(\log B)$ layers of the complete binary search tree, since a memory block of size $B$ will store a subtree of size $\Theta(\log B)$ in the layout. Thus, descending through $O(\log N)$ layers takes $O(\log N/\log B) = O(\log_B N)$.
\end{proof}

It is also useful to note that a vEB layout enables fast inorder traversal of nodes. In particular, we can state the following result, which \cite{brodal2002cache} does not even bother giving a number, and instead simply writes in the body of the text:
\begin{theorem}[\cite{brodal2002cache}, with a simplifying assumption]\label{thm:traversal}
Assuming that $M > H$, an inorder traversal starting at any node $v$ in a vEB layout takes time $O(\log_B N + S(v) / B)$, where $S(v)$ is the number of elements in the subtree at $v$.
\end{theorem}
\begin{proof}
A block of size $B$, again, stores $\Theta(\log B)$ layers of the search tree. Suppose that our inorder traversal is currently at some node $u$. Suppose we always use the cache to store the blocks that include the current path from the root down to $u$. Adding these blocks to the cache at the beginning of the traversal takes time $O(\log_B N)$. Note that this takes at most $H$ blocks of memory, and guarantees that a block that is evicted from the cache will not ever be needed again, since we're doing an in-order traversal. Further, during the traversal, a subtree of size $B$ nodes will be accessed contiguously; therefore, we will only need to load new blocks into memory every $\Theta(B)$ nodes. Adding everything up gives the desired bound.
\end{proof}
The assumption $M>H$ only breaks down for trees with $2^M$ nodes, which won't even fit into the disk of any reasonable machine, so we consider it reasonable enough.

\subsection{Dynamic Search Trees}\label{s:main-ds}

The difficulty of using the above vEB layout to make a dynamic search tree is that we cannot rebalance the tree in the same way we would rebalance, say, a red-black or splay tree: a rotation would cause entire subtrees to be moved around, and in the vEB layout, this would require physically moving the elements instead of merely rewiring pointers. In this section, we present the dynamic cache-oblivious search tree of \cite{brodal2002cache}. This search tree does not quite give the ideal runtime of $O(\log_B N)$ memory transfers on all operations, but it comes somewhat close: searches take $O(\log_B N)$, and insertions take amortized $O(\log_B N + (\log N)^2 / B)$ memory transfers. Although this is not information-theoretically ideal, this is, to our knowledge, the best-known result for a BST that has the ideal $O(N/B)$ memory transfers for a full in-order traversal. (There are other, more complex data structures that have $O(\log_B N)$ insertions and deletions, but lose this $O(N/B)$ traversal runtime). For simplicity and brevity, we neglect deletions; they can be managed by a similar rebalancing technique that is given in \cite{brodal2002cache}.

The data structure works by storing all the elements in a complete vEB layout with $2^H > N$ nodes, some of which may be empty, and periodically rebalancing subtrees. Pick a {\it density parameter} $\tau_1 \in [1/2, 1)$, and for $i = 2, \dots, H$ define $\tau_i = \tau_{i-1} + \Delta$ where $\Delta = (1 - \tau_1)/(H-1)$, so that the $\tau_i$ are evenly spaced in the interval $[\tau_1, 1]$. Let $N(v)$ be the number of nodes in the subtree rooted at $v$, $d(v)$ be the depth of node $v$, and $S(v) = 2^{H - d(v)+1}-1$ denote the size of the complete subtree rooted at $v$, {\it including} layers of the tree that may not yet be occupied by any nodes.

Insertion works as follows: compute the location $u$ at which the new key should be inserted. If $u$ is within the tree (i.e. at depth at most $H$), then place $u$ in the correct spot. Otherwise, let $v$ be the nearest ancestor of $u$ such that $N(v) < \tau_{d(v)}S(v)$, or the root of the tree if no such ancestor $v$ exists. This is essentially the closest ancestor to $u$ that has a smaller percentage of the nodes below it filled than its height $\tau_{d(v)}$ allows.

We will make space for $u$ by {\it complete rebalancing}: perform an inorder traversal of the subtree rooted at $v$, storing the elements in a separate array, and then re-inserting them, together with $u$, as compactly as possible into the subtree $v$ (i.e. taking $\log_2 N(v) + O(1)$ layers), using another inorder traversal. Notice that since $N(v) < S(v)$, there must now be space for the additional node $u$ to be added in while the elements are being inserted.

If we ever find that $N(\text{root}) \ge \tau_1 S(\text{root})$ before an insertion, then we perform a complete rebalance of the root node itself, incrementing $H$ (and re-allocating the array) between the two traversals. We now claim that this insertion method does indeed achieve the desired time bound. We will give a slightly simplified version of the analysis in \cite{brodal2002cache}. We start with a simple lemma.
\begin{lemma}[\cite{brodal2002cache}, Lemma 3.1, modified slightly]\label{lem:rebalance}
Immediately after a rebalancing of any node $v$, for any descendant $w$ of $v$, we have $N(w) < N(v) (1+ S(w))/S(v)$.
\end{lemma}
\begin{proof}
Observe that $(1 + S(w))/S(v) > (1 + S(w)) / (1 + S(v)) = 2^{d(v) - d(w)}$, and that $v$ is now perfectly balanced, so each descendant $w$ of $v$ has size at most $N(v) 2^{d(v) - d(w)}$. Substituting gives the desired result.
\end{proof}

\pagebreak[3]
\begin{theorem}[\cite{brodal2002cache}, Theorem 3.1]\label{thm:log2}
Inserting a key as described above takes amortized $O(\log_B N + (\log N)^2 / B)$ memory transfers.
\end{theorem}
\begin{proof}
Traversing the tree to find the insertion location takes time $O(\log_B N)$. It only remains to bound the work done by rebalancinng.

Suppose we rebalance node $v$, due to an insertion below $v$. Then there is some child $w$ of $v$ for which $N(w) \ge \tau_{d(w)} S(w)$ (otherwise, we would not be rebalancing $v$). Consider the last time node $v$ or any of its ancestors was rebalanced. Immediately after this rebalancing, Lemma \ref{lem:rebalance} gives that $N(w) < \tau_{d(v)} (1 + S(w)) \le \tau_{d(v)} S(w) + 1$. Thus, there must have been at least $\tau_{d(w)} S(w) - \tau_{d(v)} S(v)-1 = \Delta S(w)-1$ insertions at node $w$ or its descendants. The rebalance takes time $O(\log_B N + S(v)/B)$ by Theorem \ref{thm:traversal}, since it just consists of inorder traversals. The $O(\log_B N)$ term is accounted for by the tree traversal, and the $O(S(v)/B)$ term can be accounted for using the banker's method by charging $O(S(v)/B\Delta S(w)) = O(2/(B\Delta) = O(H/B)$ to each of the insertions under $w$ (noting that $S(v) / S(w) \approx 2 = O(1)$). Each insertion is charged at most $H$ times (once for each ancestor) between rebalances; thus, rebalancing amortizes to $O(H^2/B) = O((\log N)^2/B)$. 

Incrementing $H$ is only done every $O(N)$ insertions and consists of two tree traversals. The tree traversals, as before, take time $O(\log_B N + N/B)$, where the $\log_B N$ term is accounted for in the search, and the $N/B$ term can be (more than) accounted for by placing 1 additional credit on each insertion. Recomputing the indices of the vEB tree can be done during the in-order traversal by computing parent and child indices at each node. As argued in Section \ref{s:indices}, computing parent and child indices takes time $O(\log H) = O(\log \log N)$; thus, this amortizes to $O((\log \log N)/B)$ and does not affect the overall runtime, since that runtime already contains an $O((\log N)^2/B)$ term. Adding everything up, the overall amortized runtime is $O(\log_B N + (\log N)^2 / B)$, as desired.
\end{proof}

We note that the van Emde Boas layout is not special in the above data structure. Although we will lose the memory access bounds given by Theorem \ref{thm:log2}, we may replace the van Emde Boas layout by any other layout. In particular, in Section \ref{s:experiments}, we consider the ramifications of using a breadth-first layout instead of a vEB layout. We also notice that this data structure needs very little memory: a single array suffices to store the keys. Beyond that, we only need to store as much memory as necessary to compute the parent and child indices of any given node efficiently. For a BFS layout, this is a simple multiplication or division by 2; for the vEB layout, this can get more complex; see Section \ref{s:indices}.

\section{Experimental Results}\label{s:experiments}

\subsection{Speed Comparison}
We implement ordered dictionaries that support adding and searching for elements using splay trees, B-trees of different orders, and the structure seen in Section \ref{s:main-ds} (which we'll call the ``small tree'', since it doesn't use much memory), with both a BFS layout and a vEB layout. We find that, in practice, performance seems to be best for $\tau_1 = 1/2$, so we only include that case here. The B-tree of order 16 performed best, so we include that case in our tests, and we include the B-tree of order 2 for comparison. The B-tree of order 16 should be thought of as a cache-friendly structure that has been tuned for the block size of our particular machine. We also include {\tt std::set}, but we find that this is largely hopeless: {\tt std::set} is easily beaten in just about all tasks by just about all our structures.

\begin{figure}
    \centering
    \begin{minipage}{.5\textwidth}
        \centering
        \includegraphics[width=1\textwidth]{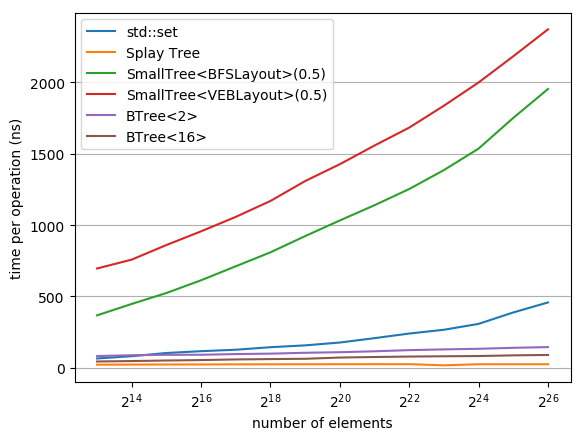} % first figure itself
        \caption{In-order insertion}
        \label{fig:beginspeed}
    \end{minipage}\hfill
    \begin{minipage}{.5\textwidth}
        \centering
        \includegraphics[width=1\textwidth]{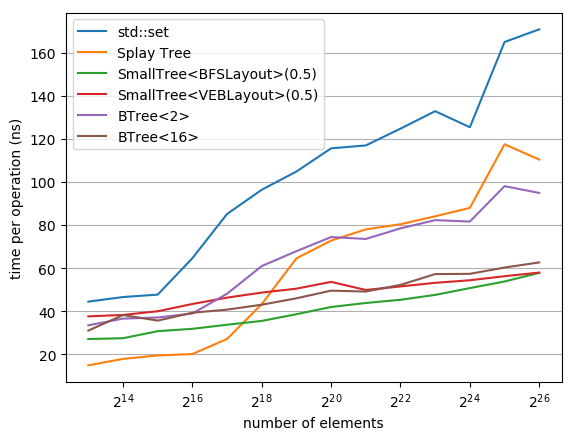} % second figure itself
        \caption{In-order traversal}
    \end{minipage}
% \end{figure}
% \begin{figure}
%     \centering
    \begin{minipage}{.5\textwidth}
        \centering
        \includegraphics[width=1\textwidth]{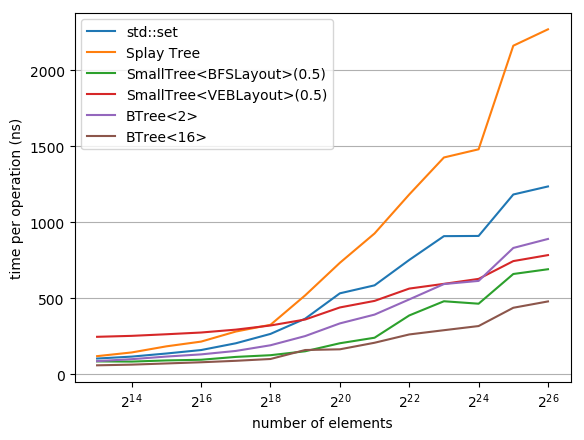} % second figure itself
        \caption{Random insertion}
    \end{minipage}\hfill
    \begin{minipage}{.5\textwidth}
        \centering
        \includegraphics[width=1\textwidth]{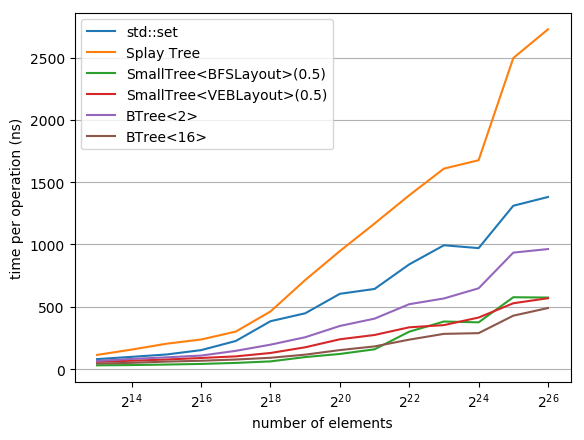} % second figure itself
        \caption{Random traversal}
        \label{fig:endspeed}
    \end{minipage}
\end{figure}

\pagebreak[3]
We run four experiments on each data structure:
\begin{enumerate}[(1)]
    \item {\it In-order insertion}: $N$ elements are inserted into the structure in sorted order. 
    \item {\it Random insertion}: $N$ elements are inserted into the structure in random order.
    \item {\it In-order traversal}: After a random insertion, every element is accessed, in sorted order.
    \item {\it Random traversal}:
    After a random insertion, every element is accessed, in a random order different from the insertion order.
\end{enumerate}

Figures \ref{fig:beginspeed} through  \ref{fig:endspeed} show the experimental results. We make several observations:
\begin{enumerate}[(1)]
    \item The in-order insertion plot clearly shows the $O((\log N)^2)$ dependence of the small tree insertion in both BFS and vEB layouts.
    \item The B-tree of order 16 is, unsurprisingly, all-around one of the fastest trees. This is unsurprising, since, as mentioned before, we can think of this tree as the explicitly-tuned ``cache-friendly'' data structure.
    \item Also unsurprisingly, the performance of the vEB layout, especially on larger sizes, is usually between that of the B-tree of order 2 (our ``standard'', unoptimized, not-cache-friendly search tree) and the B-tree of order 16 (our explicitly-tuned ``cache-friendly'' tree).
    \item The vEB layout seems to be generally slower than the BFS layout, only matching the latter's performance in traversal for very large input sizes. It is likely that there may be a crossover point past $2^{26}$ elements, but we stopped our experiments here as the runtime was becoming prohibitively expensive. The BFS layout benefits from the memory-friendliness of the small tree and the instantaneous computation of child and parent indices of any given node, but not the cache-friendless of the vEB layout. The gap can only be due to the difficulty of computing vEB indices, since the data structures are otherwise identical. See Section \ref{s:indices} for an expanded discussion of this.
    \item The splay tree performs surprisingly poorly across the board, even in inorder traversal. We suspect that this is due to poor memory locality and usage. It should be mentioned, though, that the poor inorder traversal performance disappears if the elements are inserted in sorted order; in this case, the splay tree suddenly becomes the {\it fastest} at inorder traversal. For brevity, we have not included these plots, since they look otherwise quite similar to the shown ones.
    \item {\tt std::set} performs all-around poorly. This deficiency is harder to explain.
   
\end{enumerate} 

\subsection{Calculating Indices in the van Emde Boas Layout}\label{s:indices}

Calculating the parent and child indices of an element in a van Emde Boas layout is nontrivial. Simply storing parent and child pointers/indices wastes space if it is possible to efficiently calculate these values on an on-demand basis. Common approaches involve implementing a conversion between the easy-to-compute breadth first indexing, where children are calculated with $c_1 = 2i$ and $c_2 = 2i + 1$ and parents with $p = i/2$, to van Emde Boas indexing. Many attempts to quickly calculate this conversion from breadth-first notation to van Emde Boas take time $O(\log h)$, where $h$ is the height of the tree. No fast constant-time algorithm seems to exist as of yet, and the time to compute this conversion takes longer than a cache miss does with these existing algorithms. Thus, keeping track of parent and child pointers still remains to be the quickest method, although it takes extra space within the data structure.

However, in our exploration of this topic, we attempted to improve on existing methods, seen in \cite{brodal2002cache} and implemented in \cite{coriolan} and \cite{bcopeland}. Both of these methods are sub-optimal because they either use $O(\log n)$ persistent memory, as seen in \cite{coriolan}, or $O(\log n)$ memory at runtime in the form of stack frames, as seen in \cite{bcopeland}.

While both of these methods both use $O(\log h)$ time, their use of external memory is unnecessary. We introduce a constant-memory algorithm for converting indices in breadth-first notation to the corresponding position which uses $O(\log w)$ time, where $w$ is the word size of the machine. If memory usage comes at a steep premium and cache misses are very expensive, using this algorithm with the vEB tree could be useful. The algorithm can be found in Appendix A, implemented in C++.

\begin{figure}
    \centering
    \includegraphics[scale=0.8]{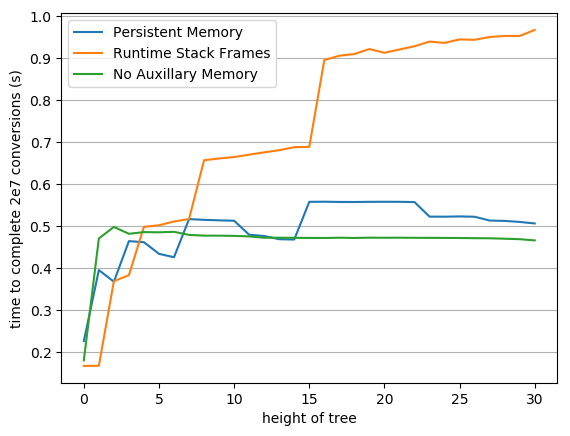}
    \caption{Runtimes of different methods of converting BFS notation to vEB notation}
    \label{fig:bfstoveb}
\end{figure}

As seen in Figure \ref{fig:bfstoveb}, our algorithm for converting BFS indices to vEB indices performs the fastest out of the three when the height is greater than seven, which corresponds to no more than 127 nodes in the tree, smaller than what is seen in a majority of applications. Even then, as discussed above, simply using a BFS layout outperforms the vEB layout due to the cost of calculating indices. Until a more efficient algorithm for calculating indices is found, it is probably best to not use the vEB layout in practice.

\section{Conclusion}
Cache-oblivious algorithms and more specifically cache-oblivious search trees are an exciting area to explore because of their relevance both in theory and in practice. In theory, using the cache-oblivious model allows us to think about how a computer uses its cache, and optimize memory usage accordingly through some very nice techniques such as the van Emde Boas layout of a binary search tree. In practice, thinking about the cache or reducing memory usage is an excellent way to improve the performance of any program; in our case, if the size of the cache is not known in advance, we find that the cache-oblivious structure (or its simpler cousin, the BFS-layout small tree) outperforms naive implementations such as the B-tree of order 2 or the splay tree. 

Finally, it is important to note the poor performance of the \texttt{std::set} class, which was in every way outperformed by almost every data structure studied in this paper. Why this is and how to possibly improve on its implementation could be an area of further study.

\bibliographystyle{plainnat}
\bibliography{main.bib}

\newpage
\section*{Appendix A} \label{appBFStovEB}
\begin{verbatim}
// BFS Index to vEB index conversion algorithm
int BFStovEB(int index, int h) {
    if (index <= 1) {
        return index - 1;
    }
    int currh = height(index);
    int remainder = index & ~(1 << currh);

    // calculate offset when h is not a power of 2
    int disttobottom = h - currh - 1;
    int thing = disttobottom ^ h;
    int afterthing = height(thing);
    int topxor = 1 << afterthing;
    int habove = h & (topxor - 1);
    int tothehabove = (1 << (h - habove));
    index = (index & (tothehabove - 1)) + tothehabove;
    currh -= habove;
    int offset = (1 << habove) - 1 +
        (remainder >> currh) * ((1 << topxor) - 1);

    // calculate spot in h = power of 2
    if (currh <= 0) {
        return offset;
    }
    int leftspine = (currh & 1) + ((currh & 2) >> 1) * 3 +
        ((currh & 4) >> 2) * 15 + ((currh & 8) >> 3) * 255 +
        ((currh & 16) >> 4) * ((1 << 16) - 1);
    remainder = index & ~(1 << currh);
    int shift = (remainder & 1) * (currh & 1) +
        ((remainder >> (currh & 1)) & 3) * ((1 << (currh & 2)) - 1) + 
        ((remainder >> (currh & 3)) & 15) * ((1 << (currh & 4)) - 1) + 
        ((remainder >> (currh & 7)) & 255) * ((1 << (currh & 8)) - 1) + 
        ((remainder >> (currh & 15)) & 0xFFFF) * ((1 << (currh & 16)) - 1);
    return offset + leftspine + shift;
}
\end{verbatim}

\end{document}